\documentclass{article}
\usepackage{ijcai26}

\usepackage{times}
\usepackage{soul}
\usepackage{url}
\usepackage[hidelinks]{hyperref}
\usepackage[utf8]{inputenc}
\usepackage[small]{caption}
\usepackage{graphicx}
\usepackage{amsmath}
\usepackage{amsthm}
\usepackage{booktabs}
\usepackage{algorithm}
\usepackage{algorithmic}
\usepackage[switch]{lineno}

\usepackage{comment}  
\nolinenumbers

\newtheorem{theorem}{Theorem}
\newtheorem{lemma}{Lemma}
\newtheorem{corollary}{Corollary}
\newtheorem{observation}{Observation}

\title{Learning-Augmented Online TRP on a Line}

\author{
    Swapnil Guragain and Gokarna Sharma
    \affiliations
    Kent State University
    \emails
    \{sguragai, gsharma2\}@kent.edu
}

\begin{document}

\maketitle

\begin{abstract}
We study the online traveling repairperson problem on a line within the recently proposed learning-augmented framework, which provides predictions on the requests to be served via machine learning.  In the original model (with no predictions), there is a stream of requests released over time along the line. The goal is to minimize the sum (or average) of the completion times of the requests. 
In the original model, the state-of-the-art  competitive ratio lower bound is $1+\sqrt{2} > 2.414$ for any deterministic  algorithm and the state-of-the-art competitive ratio upper bound is 4 for a deterministic algorithm.  Our prediction model involves predicted positions, possibly error-prone, of each request in the stream known a priori but the arrival times of requests are not known until their arrival. 
We first establish a 3-competitive lower bound which extends to the original model.
We then design a deterministic algorithm that is $(2+\sqrt{3})\approx 3.732$-competitive when predictions are perfect. With imperfect predictions (maximum error $\delta > 0$),  we show that our deterministic algorithm becomes $\min\{3.732+4\delta,4\}$-competitive, knowing $\delta$. To the best of our knowledge, these are the first results for online traveling repairperson problem in the learning-augmented framework. 
\end{abstract}






\section{Introduction}

We consider the {\em online} variant of the fundamental {\em traveling repairperson problem} (TRP). We solve online TRP on a real line $\mathcal{L}$ of length $|\mathcal{L}|$.
The input is a stream of requests released over time along the points on $\mathcal{L}$. The goal in TRP is to minimize the sum (or average) of the completion times of the requests, in contrast to {\em traveling salesperson problem} (TSP) where  the goal is to minimize the maximum completion time (or makespan).  TRP facilitates quickly serving each request rather than TSP which minimizes the maximum time to serve any request.
TRP has a variety of relevant applications in logistics and robotics.  Consider for an example a salesman/repairman/vehicle/robot (denoted as {\em server} in this paper) that has to serve locations on its workspace. 
If the server's workspace is a road/highway segment, then it resembles a line and our results can apply to such application. The algorithm we design will run by the server to serve requests. 

Online TRP has been studied quite extensively 
on the {\bf \em original} model where locations of the requests as well as their arrival times  are not known until they arrive; we will define competitive ratio formally later.  In this paper, we study online TRP  considering  a stronger {\bf\em prediction} model within the recently proposed learning-augmented framework \cite{LykourisV18,Rohatgi20,Kumar18}(Table \ref{table:model-comparison2} compares original and prediction models). The prediction model provides some knowledge on problem beforehand through some external source as considered in the learning-augmented framework. The external source can be a  machine learning approach running over the historical data relevant to the problem. Specifically, we consider the prediction model that involves predictions on the {\em locations} as considered in \cite{GouleakisLS23} where requests will arrive over time. Notice that in the original model (no predictions), no such knowledge is provided to the online algorithm $ON$. In our prediction model, $ON$ is provided beforehand with the  predicted locations of the future requests. 

\begin{table}[!t]
\footnotesize{
\centering
\begin{tabular}{ll}
\toprule
{\bf Model} & {\bf Characteristic}  \\ 
\toprule
Original &  Requests arriving at time $t$ are known at $t$ \\
\hline
Prediction & All request locations are known beforehand \\
& but not time $t$ the requests arrive\\ 
\bottomrule
\end{tabular}
\caption{Comparing original and prediction models.
\label{table:model-comparison2}
}
}
\vspace{-0mm}
\end{table}

Online TSP, where the goal is minimize the maximum completion time, has been studied on a line in the prediction model we discussed above  \cite{GouleakisLS23}. 
However, to the best of our knowledge, there is no previous study on online TRP in the prediction model even on a line. This paper is the first attempt on online TRP in the prediction model on a line as in \cite{GouleakisLS23}.

\begin{table}[!t]
\footnotesize{
\centering
\begin{tabular}{lll}
\toprule
{\bf Paper} & {\bf Competitive Ratio} & {\bf Model} \\ 
\toprule
\cite{feuerstein2001line} &  $1+\sqrt{2}\approx 2.414$ & Original \\

\bottomrule
{\bf This paper} & $3$ & {\bf Original}\\ 
\bottomrule
\end{tabular}
\caption{Deterministic online TRP lower bounds on a line.}
\label{table:model-comparison1}
}

\end{table}

\vspace{1mm}
\noindent{\bf Our Results.} 
We have the following four results. We say that the prediction is {\em perfect} if the request predicted to arrive at a location actually arrives at that location. The prediction is called {\em imperfect} otherwise with parameter $0\leq \delta\leq 1$ denoting the amount of prediction error. We say that line $\mathcal{L}$ is a {\em half-line} if servers' location is at its either end.
\begin{itemize}
\item We establish that no deterministic algorithm can achieve 3 competitive ratio on a line in our prediction model even under perfect predictions. Interestingly, this lower bound applies to the original model since our prediction model makes the problem stronger. 
({\bf Section \ref{section:lower-half}}) 

\item We develop a deterministic algorithm with a competitive ratio upper bound of  $2+\sqrt{3}\approx 3.732$ on a half-line in the original model. This bound extends to our prediction model with any prediction error $0\leq \delta\leq 1$. 
({\bf Section \ref{section:upper-half}}) 

\item In our prediction model with perfect predictions, we develop a deterministic algorithm with competitive ratio $2+\sqrt{3}\approx 3.732$ on a line. ({\bf Section \ref{section:upper-line-0}}) 

\item In our prediction model with imperfect predictions (error $\delta>0$), we give a deterministic $\min\{3.732+4\delta,4\}$-competitive algorithm, knowing $\delta$. ({\bf Section \ref{section:upper-line-d}}) 

\end{itemize}

Tables \ref{table:model-comparison1} and \ref{table:model-comparison} summarize the existing results in the original model as well as our results in the prediction model.

\vspace{2mm}
\noindent{\bf Related Work.}
All existing results on online TRP are in the original model.
\cite{feuerstein2001line} established the only previously known competitive lower bound of $1+\sqrt{2}\approx 2.414$ for online TRP, which is the state-of-the-art. This lower bound applies to any metric that satisfies triangle inequality (line is a special case). 
We provide in this paper the improved competitive lower bound of 3 in a line (more precisely a half-line) which will also serve as a lower bound for any metric satisfying triangle inequality.

Regarding the upper bound, 
\cite{feuerstein2001line} provided the first solution to online TRP on a line, which is a  $9$-competitive deterministic algorithm.  
\cite{KrumkePPS03-TCS} obtained a $(1+\sqrt{2})^2\approx 5.8285$ in any metric which applies to line. 
\cite{BienkowskiL19-MFCS} provided a $5.429$-competitive deterministic algorithm on a line. \cite{Hwang}  give a $5.14$-competitive algorithm in any metric which applies to a line. The state-of-the-art is the 4-competitive deterministic algorithm in any metric due to  \cite{BienkowskiKL21-ICALP}. 


\section{Model, Definitions, and Notations}

\noindent{\bf Model.} We consider a single server initially located at a fixed distinguished location on $\mathcal{L}$, which we call {\em origin} $o$.  
The execution starts at time $0$ and time proceeds in discrete time steps. The requests arrive at discrete time steps on any location on $\mathcal{L}$. Any number of requests may arrive at any time step. 
%
Each request 
is communicated to the server 
as soon as it arrives, i.e., a request $r_i$ with {\em arrival time}  $t_i$ coming at  some location $p_i$ on $\mathcal{L}$ is known by the server at time $t_i$. 
Server serves requests traversing $\mathcal{L}$ starting from origin $o$. 
Request $r_i$ is considered served when the server reaches the location $p_i$ the first time while traversing $\mathcal{L}$. If server reaches $p_i$ the first time at time $t_i'$, then $t_i'$ is the {\em completion time} for $r_i$. 

\vspace{1mm}
\noindent{\bf Performance Metric -- Sum of Completion Times.}
We consider $n\geq 1$ requests ($n$ is not known to the server a priori) $\mathcal{R}:=\{r_1,\ldots,r_n\}$. Let the server runs the online algorithm $ON$ to decide on how to serve the requests. Let $ON(\mathcal{R})=\sum_{i=1}^{n} t_i'$ denote the sum of the completion times of the requests in $\mathcal{R}$. 

\vspace{1mm}
\noindent{\bf Competitive Ratio.}
We measure the efficiency of $ON$ by comparing its performance with the performance of the optimal offline algorithm $OPT$ provided all the requests in $\mathcal{R}$ arrive at time $t=0$. Let $OPT(\mathcal{R})$ be the sum of completion times   of $OPT$. We have that $OPT(\mathcal{R})\geq |OPT\_TOUR|$, where $OPT\_TOUR$ is the tour computed by $OPT$ that minimizes the completion time for each request in $\mathcal{R}$.
In online TRP, request $r_i$ with arrival time $t_i$ cannot be served before $t_i$, i.e., $OPT(\mathcal{R})\geq \sum_{i=1}^n t_i$.  
Therefore, the goal is to minimize the ratio 
$\frac{ON(\mathcal{R})}{OPT(\mathcal{R})}=\frac{\sum_{i=1}^n t_i'}{\max\{|OPT\_Tour|,\sum_{i=1}^n t_i\}}$
which is called the {\em competitive ratio} in the field of online algorithm design \cite{Borodin98}. The best possible ratio is 1. The above ratio can be written for each request $r_i$ as  
$\frac{ON(r_i)}{OPT(r_i)}=\frac{t_i'}{\max\{|OPT\_Tour_i|,t_i\}},$
where $|OPT\_TOUR_i|$ is the tour length of $OPT\_TOUR$  until $r_i$ is served. Minimizing $\frac{ON(r_i)}{OPT(r_i)}$ for each individual request $r_i$ minimizes $\frac{ON(\mathcal{R})}{OPT(\mathcal{R})}$ for all requests in $\mathcal{R}$. Therefore, in this paper, we strive to minimize individual ratio $\frac{ON(r_i)}{OPT(r_i)}$. 

\begin{table}[!t]
\footnotesize{
\centering
\begin{tabular}{lll}
\toprule
{\bf Algorithm} & {\bf Competitive Ratio} & {\bf Model} \\ 
\toprule
\cite{feuerstein2001line} &  9 & Original \\
\hline
\cite{KrumkePPS03-TCS} & $3+2\sqrt{2}\approx 5.83$ & Original\\
\hline
\cite{BienkowskiL19-MFCS} & $5.429$ & Original\\
\hline
\cite{Hwang} & $5.14$  & Original\\
\hline
\cite{BienkowskiKL21-ICALP}  & 4 & Original\\ 
\bottomrule
{\bf This paper, $\delta=0$} & $2+\sqrt{3}\approx 3.732$ & {\bf Prediction}\\
\hline
{\bf This paper, known $\delta>0$} & $\min\{3.732+4\delta,4\}$ & {\bf Prediction}\\ 
\bottomrule
\end{tabular}
\caption{Deterministic online TRP solutions on a line.}
\label{table:model-comparison}
}

\end{table}

\vspace{1mm}
\noindent{\bf Line and Half-Line.}
Let the endpoints of $\mathcal{L}$ be $a,b$. 
The origin  $o\in \mathcal{L}$. We measure distance on $\mathcal{L}$ from $o$, i.e., point $a$ ($b$) is at distance $|a|$ ($|b|$). Let us denote by $\mathcal{L}_{right}$ ($\mathcal{L}_{left}$) the part of $\mathcal{L}$ that is on the right (left) of $o$. Let $x$ ($y$) be a point on $\mathcal{L}_{left}$ ($\mathcal{L}_{right}$). We denote by $[x,y]$ the segment of $\mathcal{L}$ between $x$ and $y$ (inclusive).  
We call $\mathcal{L}$ half-line when origin $o$ is either $a$ or $b$, that is, $\mathcal{L}$ is either only $\mathcal{L}_{right}$ or only $\mathcal{L}_{left}$.
%
%

\vspace{1mm}
\noindent{\bf Prediction Model.}
We denote a request $r_i\in \mathcal{R}$ by a triple $(a_i^{loc},p_i^{loc},t_i)$, where $a_i^{loc}$ is the actual point on $\mathcal{L}$ it arrives, $p_i^{loc}$ is the predicted point on $\mathcal{L}$ it arrives, and $t_i$ is the actual time it arrives. 
In our prediction model, the predicted location $p_i^{loc}$ of each request $r_i$ is known beforehand but not $a_i^{loc}$ and $t_i$ until $r_i$ arrives at $t_i$. With {\em prefect} predictions, $a_i^{loc}=p_i^{loc}$, i.e., $r_i$ arrives at the predicted location. In the original model, there is no notion of $p_i^{loc}$ known beforehand, i.e., request $r_i$ is known only when it comes at $t_i$ at location $a_i^{loc}$. 

\vspace{1mm}
\noindent{\bf Prediction Error.}
We denote by $\Delta$ the error in prediction. It is measured based on predicted and actual locations of each request $r_i$. For $r_i$, the error on prediction $\Delta_i=|a_i^{loc}-p_i^{loc}|$. For $\mathcal{R}$, we denote by $\Delta$ the maximum error on prediction, i.e., $\Delta:=\max_{i=1}^{n} \Delta_i$.  Notice that $\Delta\leq |\mathcal{L}|$, i.e., the error on prediction cannot be more than $|\mathcal{L}|$ since each request in $\mathcal{R}$ arrive on $\mathcal{L}$ ($a,b$ inclusive). Therefore, we can use $\delta$ as the error ratio parameter giving the amount of error w.r.t. $|\mathcal{L}|$, i.e., $0\leq \delta=\frac{\Delta}{|\mathcal{L}|}\leq 1$. $\delta=0$ denotes perfect prediction and $\delta=1$ denotes worst-case prediction error. Given $\delta$, $\Delta$ can be computed, and vice-versa, given known $|\mathcal{L}|$.

\begin{figure}[t]
    \centering
    \includegraphics[width=\linewidth]{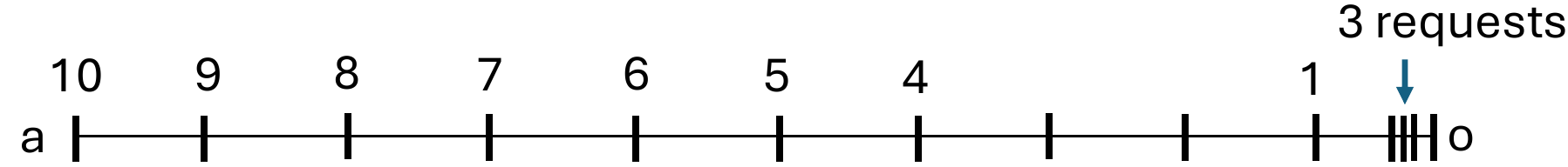}
    \caption{The lower bound construction on a half-line $\mathcal{L}$ with $|\mathcal{L}|=10$. The request locations are shown. The 8 requests at 1, 4--10 arrive at time $t=0$. The 3 remaining requests  arrive over time. }
    \label{fig:lower}
\end{figure}

\section{3-Competitive Lower Bound}
\label{section:lower-half}
Our 3-competitive lower bound is strong compared to the state-of-the-art $1+\sqrt{2}\approx 2.414$ \cite{feuerstein2001line}. We establish our lower bound
considering $\mathcal{L}$ being half-line, and hence applies to any metric satisfying triangle inequality. We construct an adversarial input $\mathcal{R}$ on half-line $\mathcal{L}$ such that no deterministic algorithm can satisfy  ratio $\frac{ON(r_i)}{OPT(r_i)}$ less than 3 for  at least a request $r_i$ in the prediction model even with perfect predictions.
Since the original model does not have predictions, the lower bound applies to the original model. 
In particular, we prove the following theorem.
\begin{theorem}[{\bf lower bound}]
    The exists an input $\mathcal{R}$ for which no deterministic algorithm can achieve 3-competitive ratio for online TRP in the  prediction model. 
\end{theorem}
\begin{proof}
    For a deterministic algorithm $ON$ to be 3-competitive, we need to guarantee 
    $\frac{ON(r_i)}{OPT(r_i)}\leq 3$ for each request $r_i\in \mathcal{R}$.
    Suppose $r_i$ arrives at time $t_i$ on location $p_i\in \mathcal{L}$ with distance $l_i$ from $o$.   
    $OPT(r_i)\geq \max\{t_i,l_i\}$.

We consider 11 requests in $\mathcal{R}$ with 8 predicted locations at distance  $1,4,5,\ldots,{10}$ from $o$ and 3 predicted locations  at distance $\epsilon_i<<1$ from $o$. We assume that the locations are known at time $t=0$ (see Fig.~\ref{fig:lower}). 

Let $ON$ be a deterministic algorithm to serve the requests that arrive on the predicted locations.  We assume that, at time step $t=0$,  a request each arrives at $1,4,5,\ldots,10$.  
We denote the requests by $r_i$, where $i$ is the location, i.e., for the request arriving at 4, we denote it by $r_4$.
We will release requests at locations $\epsilon_i$ over time. For $ON$ to guarantee 3-competitive ratio,  the requests at $1,4,5,\ldots,{10}$ must be served by time step $3, 12, 15, 18, 21, 24, 27, 30$. 
We will show that there is at least a request  among $\{
r_1,r_4,r_5,\ldots,r_{10}\}$ for which $\frac{ON(r_i)}{OPT(r_i)}> 3$ for any deterministic algorithm $ON$. 

We consider four cases, Case 1--4, based on  whether $ON$ waits at origin $o$ and, if so, how many time steps.


\begin{itemize}

\item {\bf Case 1. $ON$ does not wait at $o$:} $ON$ serves $r_1$ (the request at distance 1) at step 1. We release $r_{\epsilon_1}$ at step 1. To have $\frac{ON(r_{\epsilon_1})}{OPT(r_{\epsilon_1})}\leq  3$, $r_{\epsilon_1}$ needs to be served by step 3. Since $ON$ is at $1$, and need to reach $\epsilon_1$ by step 3, we have the following. 
\begin{itemize}
\item {\bf Case 1.A $ON$ serves $r_{\epsilon_1}$ at step 2:}
After serving $r_{\epsilon_1}$, $OL$ has two choices: stay at ${\epsilon_1}$ or leave.
Until $ON$ stays at  ${\epsilon_1}$, we do not release any new request. Let $s>2$ be the step at which $ON$ reaches $1$ leaving ${\epsilon_1}$. We release $r_{\epsilon_2}$ at step $s$. Notice that $s$ must be $\leq 9$ since otherwise $r_4$ is not reached by step 12. $ON$ has two choices, 
\begin{itemize}
\item {\bf Case 1.A.i $ON$ serves  $r_{\epsilon_2}$ at step $s+1$}
We do not release any new request until $ON$ stays at  ${\epsilon_2}$ after serving $r_{\epsilon_2}$. Let $s'>s+2$ be the step at which $ON$ reaches $1$ leaving ${\epsilon_2}$. We release $r_{\epsilon_3}$ at step $s'$. Notice that $s'$ must be $\leq 9$ since otherwise $r_4$ is not reached by step 12.
\item {\bf Case 1.A.ii $ON$ serves $r_4,r_5,\ldots$ and returns to $r_{\epsilon_2}$ by step $3s$.} Since $s,s'\leq 9$, $ON$ can only visit up to $r_s$ ($r_{s'}$) to return to ${\epsilon_2}$ by step $3s$ ($3s'$). If doing so, $r_{s+1}$ ($r_{s'+1}$) takes at least $3s+s+1$ ($3s'+s'+1$) steps to be served, giving $\frac{ON(r_{s+1})}{OPT(r_{s+1})}> 3$ ($\frac{ON(r_{s'+1})}{OPT(r_{s'+1})}>3$). If $ON$ returns after visiting $r_{s+1}$ ($r_{s'+1}$) , then $r_{\epsilon_2}$ is not served before step $3s+1$ ($3s'+1$), giving $\frac{ON(r_{\epsilon_2})}{OPT(r_{\epsilon_2})}> 3$.   
\end{itemize}
\item {\bf Case 1.B $ON$ waits at $1$ until step 2 and serves $r_{\epsilon_1}$ at step 3:} After serving $r_{\epsilon_1}$, $ON$ has two choices as in Case 1.A. Until $ON$ stays at ${\epsilon_1}$, we do not release any new request. We know that $ON$ must leave ${\epsilon_1}$ and reach $1$ by step $4\leq s\leq 9$ to serve $r_4$ by time step 12. We release $r_{\epsilon_2}$ at step $s$. Then either Case 1.A.i or Case 1.A.ii applies. 
\item {\bf Case 1.C $ON$ waits at $1$ until step 3 or $ON$ goes to $2$ at step 2 and returns to ${\epsilon_1}$ by step 4:} We have that $\frac{ON(r_1)}{OPT(r_1)}> 3$ and we are done. 
\end{itemize}

\item {\bf Case 2 $ON$ waits at $o$ for 1 time step:} $r_1$ is served at step 2. We release $r_{\epsilon_1}$ at step 2. $ON$ must serve $r_{\epsilon_1}$ by step 6 to guarantee $\frac{ON(r_{\epsilon_1})}{OPT(r_{\epsilon_1})}\leq  3$.  $ON$ has the following choices:
\begin{itemize}
\item {\bf Case 2.A $ON$ serves $r_{\epsilon_1}$ at step 3, 4, 5, or 6 (waiting 0, 1, 2, 3 steps at $1$):} We do not release any new request until $r_{\epsilon_1}$ is served and $ON$ returns to $1$. Suppose $ON$ returns to $1$ at $4\leq s\leq 7$ after serving  $r_{\epsilon_1}$ in any step $[3,s-1]$. We release $r_{\epsilon_2}$ at step $s$ which must be served by $3s$. Again $ON$ has two choices.
\begin{itemize}
\item {\bf Case 2.A.i $ON$ serves  $r_{\epsilon_2}$ at step $s+1$:}
We do not release any new request until $ON$ stays at  ${\epsilon_2}$ after serving $r_{\epsilon_2}$. Let $s'\geq s+2$ be the step at which $ON$ reaches $1$ leaving ${\epsilon_2}$. We release $r_{\epsilon_3}$ at step $s'$. Notice that $s' \leq 9$ since otherwise $r_4$ is not reached by step 12.
\item {\bf Case 2.A.ii $ON$ serves $r_4,r_5,\ldots$ and returns to ${\epsilon_2}$ by step $3s$.} Since $s,s'\leq 7$, $ON$ can only visit up to $r_s$ ($r_{s'}$) to return to ${\epsilon_2}$ by step $3s$ ($3s'$). If doing so, $r_{s+1}$ ($r_{s'+1}$) takes at least $3s+s+1$ ($3s'+s'+1$) steps to be served, giving $\frac{ON(r_{s+1})}{OPT(r_{s+1})}> 3$ ($\frac{ON(r_{s'+1})}{OPT(r_{s'+1})}>3$). If $ON$ returns after visiting $r_{s+1}$ ($r_{s'+1}$), then $r_{\epsilon_2}$ is not served before step $3s+1$ ($3s'+1$), giving $\frac{ON(r_{\epsilon_2})}{OPT(r_{\epsilon_2})}> 3$.   
\end{itemize}

\item {\bf Case 2.B $ON$ attempts to serve $r_4$ and then $r_{\epsilon_1}$:}   $r_{\epsilon_1}$ cannot be served by step 6 if $r_4$ is served. 
\end{itemize}
\item {\bf Case 3 $ON$ waits at $o$ for 2 steps:} $r_1$ is served at step 3.  We release $r_{\epsilon_1}$ at step 3  which needs to be served by step 9 for $\frac{ON(r_{\epsilon_1})}{OPT(r_{\epsilon_1})}\leq  3$. 
\begin{itemize}
\item {\bf Case 3.A $ON$ serves $r_{\epsilon_1}$ at step 4, \ldots, 8 (waiting 0, \ldots, 4 steps at $1$):} We do not release any new request until $r_{\epsilon_1}$ is served and $ON$ returns to $1$. Suppose $ON$ returns to $1$ at $5\leq s\leq 9$ after serving  $r_{\epsilon_1}$ in any step $[4,s-1]$. We release $r_{\epsilon_2}$ at step $s$ which must be served by $3s$. $ON$ again has two choices.
\begin{itemize}
\item {\bf Case 3.A.i. Similar to Case 2.A.i $ON$ serves $r_{\epsilon_2}$ at step $s+1$:} Let $s'\geq s+2$  be the step at which $ON$ is at $1$ after serving $r_{\epsilon_2}$. $ON$ must reach $1$ by $s'\leq 9$ since otherwise $r_4$ is not reached by step 12.
\item {\bf Case 3.A.ii $ON$ serves $r_4,r_5,\ldots$ and returns to ${\epsilon_2}$ by step $3s$.} Since $s,s'\leq 9$, $ON$ can only visit up to $s$ ($s'$) to return to ${\epsilon_2}$ by step $3s$ ($3s'$). If doing so, we reach Case 2.A.ii.
\end{itemize}

\item {\bf Case 3.B Attempt to serve $r_4$ and then $r_{\epsilon_1}$:}   $r_{\epsilon_1}$ cannot be served by step 6 if $r_4$ is served. 
\end{itemize}
\item {\bf Case 4. $OL$ waits at $o$ for 3 steps or more:} We have that $\frac{ON(r_1)}{OPT(r_1)}> 3$ and we are done. 
\end{itemize}
Therefore, there is a request in $\mathcal{R}$ which cannot be served with competitive ratio $\leq 3$ by any deterministic algorithm. 
\end{proof}

We have the following corollary for the lower bound on the original model.  The proof is immediate since above lower bound is in the prediction model and the prediction model is stronger than the original model. 
\begin{corollary}
    No deterministic algorithm can achieve 3-competitive ratio for online TRP in the original model.
\end{corollary}

\section{$(2+\sqrt{3})$-Competitive Deterministic Algorithm on a Half-Line}
\label{section:upper-half}

We develop here a $2+\sqrt{3}\approx 3.732$-competitive deterministic algorithm for online TRP on a half-line $\mathcal{L}$. 
The algorithm works in both original and prediction models.  This algorithm will be a  fundamental building block for our line algorithm presented in the next section. 

\begin{algorithm}[t!]
\footnotesize
\begin{algorithmic}[1]
\STATE starting at time step $0$, server traverses $\mathcal{L}$ in round trips   $RT_j,j\geq 1,$ starting from and ending at origin $o$ until there is at least an outstanding request.
\IF {$j=1$} 
\STATE visit segment of length $\frac{2+2\alpha}{2}$ from $o$.
\ELSE
\STATE visit segment of length $\frac{(2+2\alpha)^{j-1}(1+2\alpha)}{2}$ from $o$. 
\ENDIF
\caption{Algorithm on a half-line  $\mathcal{L}$}
 \label{algorithm:half-line}
 \end{algorithmic}
\end{algorithm}

\begin{figure}[t!]
    \centering
    \includegraphics[width=\linewidth]{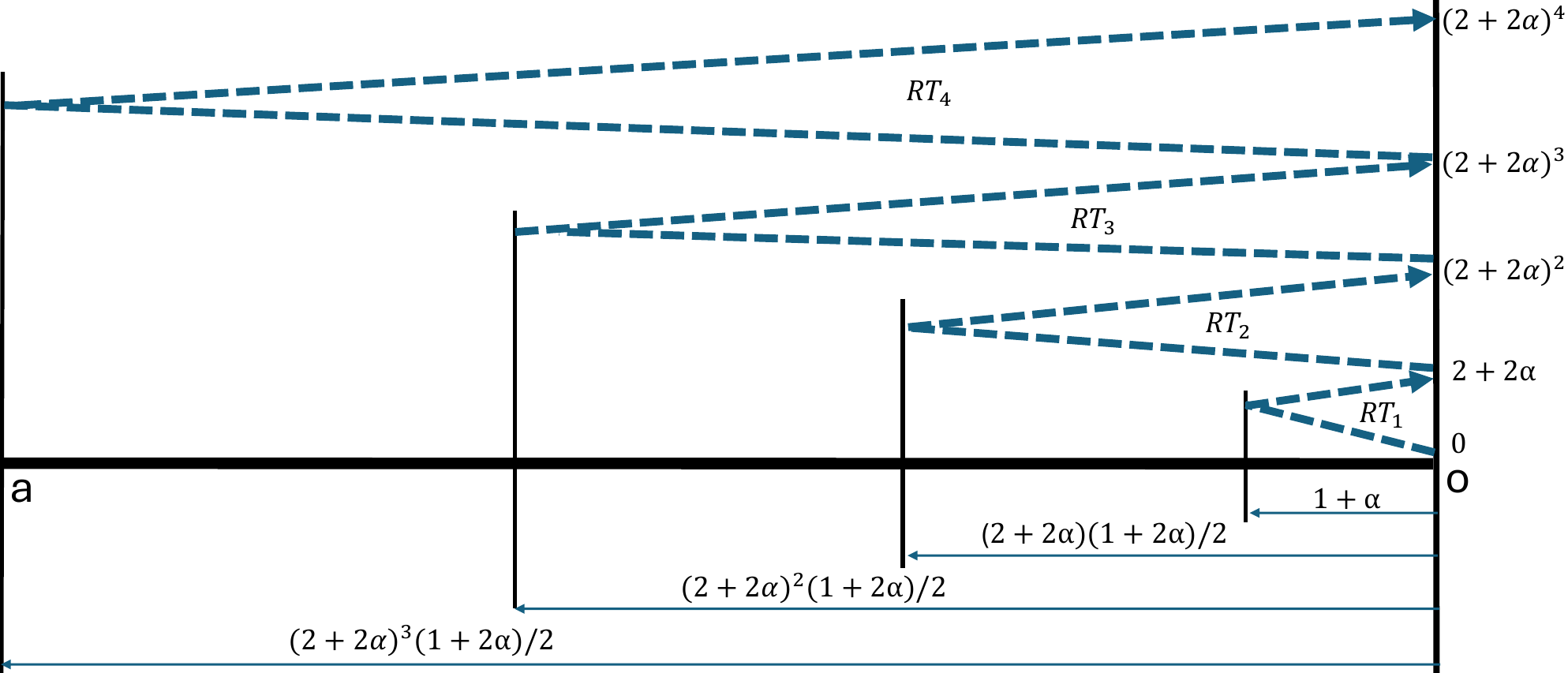}
    \caption{An illustration of how Algorithm \ref{algorithm:half-line} visits half-line $\mathcal{L}$ in round-trips $RT_j, j\geq 1,$ starting from and ending at origin $o$. The segment lengths of $\mathcal{L}$ visited in each $RT_j$ is also shown at bottom.}
    \label{fig:halfline}
\end{figure}

Our algorithm $ON$ is outlined in Algorithm \ref{algorithm:half-line} and illustrated in Fig.~\ref{fig:halfline}. Server is initially at origin $o$, which is the one end of $\mathcal{L}$ (and hence the half-line). Starting at time step $0$, server visits $\mathcal{L}$ in round trips. In the first round trip $RT_1$, it visits the segment of $\mathcal{L}$ of length $\frac{2+2\alpha}{2}$ and returns to $o$, where $\alpha$ is some constant (we set its value later). Therefore, the length of $RT_1$ is $|RT_1|=2+2\alpha$. 
For $RT_j,j\geq 2$,  it visits the segment of length $\frac{(2+2\alpha)^{j-1}(1+2\alpha)}{2}$ with $|RT_j|=(2+2\alpha)^{j-1}(1+2\alpha)$.
The server stops 
as soon as  there is no outstanding request. 

We now analyze Algorithm \ref{algorithm:half-line} for its correctness and competitive ratio guarantees. Correctness means that the server eventually serves all the requests no matter when and where that arrive. This is immediate since each round trip is of increasing length starting and ending at $o$ and hence all the request locations on $\mathcal{L}$ are eventually visited. For the competitive ratio, we establish the following theorem.

\begin{theorem}
\label{theorem:half-line}
Algorithm \ref{algorithm:half-line} is $\max\{2+2\alpha, \frac{5+6\alpha}{1+2\alpha}\}$-competitive for online TRP on a half-line $\mathcal{L}$ in the original model. 
\end{theorem}
\begin{proof}
%

Let $r_i=(l_i,t_i)$ be the request that arrives at distance $l_i$ from $o$ at time step $t_i\geq 0$. We have that $OPT(r_i)\geq \max\{l_i,t_i\}$. Let $ON(r_i)$ be the time at which $ON$ (Algorithm \ref{algorithm:half-line}) serves $r_i$.
We show that $ON(r_i)\leq \max\{2+2\alpha, \frac{5+6\alpha}{1+2\alpha}\} \cdot OPT(r_i).$

We consider two cases. The first case deals with $l_i\geq t_i$ and the second case deals with $l_i<t_i$.  

{\bf Case 1: $l_i\geq t_i$:} In this case, $OPT(r_i)\geq l_i$. Since $t_i\leq l_i,$ $l_i\geq t_i\geq 0$. We prove here that $ON(r_i)\leq 
\frac{5+6\alpha}{1+2\alpha} 
OPT(r_i).$ 

We consider three cases of $l_i$ separately:

\[
  l_i =
  \begin{cases}
    [0,1+\alpha] & 
    \textbf{Case 1.1}\\
    ((1+\alpha), \frac{(2+2\alpha) (1+2\alpha)}{2}] & 
    \textbf{Case 1.2}\\
    (\frac{(2+2\alpha)^{j-2}(1+2\alpha)}{2},\frac{(2+2\alpha)^{j-1}(1+2\alpha)}{2}] & 
    \textbf{Case 1.3}\\
  \end{cases}
\]

{\bf Case 1.1:} Consider first the case of $l_i\in [0,(1+\alpha)]$. 
Since $l_i\geq t_i$, $r_i$ must have arrived by the time $ON$ reaches $l_i$ during the first round-trip $RT_1$. That is, $r_i$ must have arrived at time step $t_i\leq \lfloor l_i\rfloor$, otherwise $l_i\notin [0,(1+\alpha)]$.  Therefore $r_i$ must be served by $ON$ while going from $o$ to $1+\alpha$ giving competitive ratio 1.

{\bf Case 1.2:} Now consider the second case of $l_i\in  ((1+\alpha), \frac{(2+2\alpha) (1+2\alpha)}{2}]$.
$OL$ visits $l_i$ for the first time in the second round-trip $RT_2$ since $t_i\leq l_i$. This gives $ON(r_i)\leq l_i+|RT_1|\leq l_i+(2+2\alpha)$. 
Therefore, $\frac{OL(r_i)}{OPT(r_i)}=\frac{l_i+(2+2\alpha)}{l_i}.$
The value of $\frac{ON(r_i)}{OPT(r_i)}$ is maximized picking $l_i=(1+\alpha)$ when $l_i\in  ((1+\alpha), \frac{(2+2\alpha) (1+2\alpha)}{2}].$
Therefore, $$\frac{ON(r_i)}{OPT(r_i)}=\frac{(1+\alpha)+(2+2\alpha)}{(1+\alpha)}\leq \frac{3(1+\alpha)}{(1+\alpha)}\leq 3.$$

{\bf Case 1.3:} We now consider the third and final case of  $l_i\in  (\frac{(2+2\alpha)^{j-2}(1+2\alpha)}{2},\frac{(2+2\alpha)^{j-1}(1+2\alpha)}{2}]$, which happens in $RT_j, j\geq 3$.
$ON$ visits $l_i$ for the first time in $RT_j$ since $t_i\leq l_i$. This gives $ON(r_i)\leq l_i+\sum_{k=1}^{j-1}(|RT_{k}|)\leq l_i+(2+2\alpha)^{j-1}$. 
Thus, $\frac{ON(r_i)}{OPT(r_i)}=\frac{l_i+(2+2\alpha)^{j-1}}{l_i}.$
Again, the value of $\frac{ON(r_i)}{OPT(r_i)}$ is maximized picking $l_i=\frac{(2+2\alpha)^{j-2}(1+2\alpha)}{2}$ when $l_i\in  (\frac{(2+2\alpha)^{j-2}(1+2\alpha)}{2},\frac{(2+2\alpha)^{j-1}(1+2\alpha)}{2}].$
Therefore, $$\frac{ON(r_i)}{OPT(r_i)}=\frac{\frac{(2+2\alpha)^{j-2}(1+2\alpha)}{2}+(2+2\alpha)^{j-1}}{\frac{(2+2\alpha)^{j-2}(1+2\alpha)}{2}}\leq \frac{5+6\alpha}{1+2\alpha}.$$

{\bf Case 2: $l_i< t_i$:} In this case, $OPT(r_i)\geq t_i$. Since $t_i> l_i,$ $t_i\geq l_i\geq 0$. We prove that $ON(r_i)\leq \max\{2+2\alpha, \frac{5+6\alpha}{1+2\alpha}\} OPT(r_i).$ 

We now consider cases of $t_i$ separately:

\[
  t_i =
  \begin{cases}
    [0,1+\alpha] & 
    \textbf{Case 2.1} 
    \\
    ((1+\alpha), (2+2\alpha)] & 
    \textbf{Case 2.2}
    \\
    ((2+2\alpha)^{j-1}, \frac{(2+2\alpha)^{j-1}(3+2\alpha)}{2} ]  & 
    \textbf{Case 2.3}\\
    (\frac{(2+2\alpha)^{j-1}(3+2\alpha)}{2},(2+2\alpha)^j]  & 
    \textbf{Case 2.4} 
  \end{cases}
\]

{\bf Case 2.1:} For $t_i \in [0,1+\alpha]$, 
$r_i$ must arrive by step $\lfloor 1+\alpha \rfloor$. $ON$ is at location $1+\alpha$ from $o$ at step $\lfloor 1+\alpha \rfloor$. 
If $r_i$'s location is between $o$ and $1+\alpha$,  we have two cases. Either $t_i=0$ or $t_i\geq 1$. 
If $t_i=0$ and $l_i=0$, $r_i$ is served immediately at step $0$. If $t_i\geq 1$, then it will be served by time step $(2+2\alpha)$. Therefore, $ON(r_i)\leq (2+2\alpha) \cdot OPT(r_i)$. 
If $r_i$'s location is not between between $o$ and $1+\alpha$, {\bf Case 1} applies since $l_i\geq t_i$, giving the competitive ratio $ON(r_i)\leq \frac{5+6\alpha}{1+2\alpha} OPT(r_i)$.  

{\bf Case 2.2:} For $t_i \in ((1+\alpha), (2+2\alpha)]$, 
$r_i$ must arrive by $(2+2\alpha)$.
If $r_i$'s location is between $o$ and $1+\alpha$ (inclusive), then it will be served by $t_i+2+2\alpha$. We have that $OPT(r_i)\geq t_i\geq  (1+\alpha)$. Therefore, $ON(r_i)/OPT(r_i)\leq 3$.
If $r_i$'s location is between $1+\alpha$ and $2+2\alpha$ from $o$, we have that $r_i$ is arrived by $2+2\alpha$. The server may be at $o$ when $r_i$ arrives. Additionally, $t_i>l_i$. Therefore, $ON$ can reach $r_i$'s location from $o$ traveling at most distance $t_i$. 
Thus, $r_i$ will be served by $ON$ in the second round-trip $RT_2$ by time $t_i+(2+2\alpha)$. 
If $r_i$'s location is not between $o$ and $2+2\alpha$ (inclusive),
{\bf Case 1} applies since $l_i\geq t_i$, giving the competitive ratio $ON(r_i)\leq \frac{5+6\alpha}{1+2\alpha} OPT(r_i)$.  

{\bf Case 2.3:}
For $t_i \in ((2+2\alpha)^{j-1}, \frac{(2+2\alpha)^{j-1}(3+2\alpha)}{2} ]$,
$r_i$ must arrive by time step $\lfloor \frac{(2+2\alpha)^{j-1}(3+2\alpha)}{2}\rfloor$. $ON$ is at $\frac{(2+2\alpha)^{j-1}(3+2\alpha)}{2}$ at time step $\lfloor \frac{(2+2\alpha)^{j-1}(3+2\alpha)}{2} \rfloor$. 
If $r_i$'s location is between $o$ and $\frac{(2+2\alpha)^{j-1}(3+2\alpha)}{2}$,   it will be served by time $(2+2\alpha)^j$. 
Additionally,
$OPT(r_i)\geq (2+2\alpha)^{j-1}.$
Therefore, $ON(r_i)\leq (2+2\alpha) \cdot OPT(r_i)$. 
If $r_i$'s location is not between $o$ and $\frac{(2+2\alpha)^{j-1}(3+2\alpha)}{2}$, we have that $l_i\geq t_i$, and {\bf Case 1} applies, giving the competitive ratio $ON(r_i)\leq \frac{5+6\alpha}{1+2\alpha} OPT(r_i)$. 

{\bf Case 2.4:} This case is similar to Case 2.2 and the competitive ratio becomes either $2+2\alpha$ or $\frac{5+6\alpha}{1+2\alpha}.$

Combining all the competitive ratios, we have the claimed $\max\{2+2\alpha, \frac{5+6\alpha}{1+2\alpha}\}$ competitive ratio for Algorithm \ref{algorithm:half-line} for online TRP on a half-line $\mathcal{L}$. 
\end{proof}

\begin{corollary}
Setting $\alpha=\frac{\sqrt{3}}{2}$, Algorithm \ref{algorithm:half-line} becomes $(2+\sqrt{3})\approx 3.732$-competitive for online TRP on a half-line $\mathcal{L}$.  
\end{corollary}
\begin{proof}
    The competitive ratio is minimized when 
    \begin{alignat*}{2}
  && 2+2\alpha &=\frac{5+6\alpha}{1+2\alpha}\notag\\
  &\Rightarrow\quad
  &(2+2\alpha)(1+2\alpha)&=5+6\alpha \notag\\
  &\Rightarrow\quad
  &2+4\alpha+2\alpha+4\alpha^2&=5+6\alpha \notag\\
  &\Rightarrow\quad
  &2+6\alpha+4\alpha^2-5-6\alpha &=0 \notag\\
  &\Rightarrow\quad
  &4\alpha^2-3 &=0 \notag\\
  &\Rightarrow\quad
  &4\alpha^2 &=3 \notag\\
  &\Rightarrow\quad
  &\alpha^2 &=\frac{3}{4} \notag\\
  &\Rightarrow\quad
  &\alpha &=\frac{\sqrt{3}}{2}\notag\\
\end{alignat*}
Therefore, setting $\alpha=\frac{\sqrt{3}}{2}$, $2+2\alpha=2+2\frac{\sqrt{3}}{2}=2+\sqrt{3}$ and 
$\frac{5+6\alpha}{1+2\alpha}=\frac{5+6\frac{\sqrt{3}}{2}}{1+2\frac{\sqrt{3}}{2}}=\frac{5+3\sqrt{3}}{1+\sqrt{3}}=2+\sqrt{3},$
since $5+3\sqrt{3}=(2+\sqrt{3})(1+\sqrt{3})$.

That is, Algorithm \ref{algorithm:half-line} is $(2+\sqrt{3})\approx 3.732$-competitive. 
\end{proof}

\section{$(2+\sqrt{3})$-Competitive Algorithm on a Line with Perfect Predictions}
\label{section:upper-line-0}




\vspace{2mm}
\noindent{\bf Highlevel Overview.}
We build our algorithm on two ideas. The first idea is Afrati {\it et al.}'s  algorithm \cite{afrati1986complexity} (which we call $AFRATI\_ALG$)  for offline TRP. 
In offline TRP, the request set $\mathcal{R}$ is known to the algorithm beforehand. 
Knowing $\mathcal{R}$, $AFRATI\_ALG$ computes a tour $OFF\_TOUR$ on $\mathcal{L}$ to visit the request locations in $\mathcal{R}$ (Fig.~\ref{fig:tour}). 
%
\begin{figure}[t!]
    \centering
    \includegraphics[width=\linewidth]{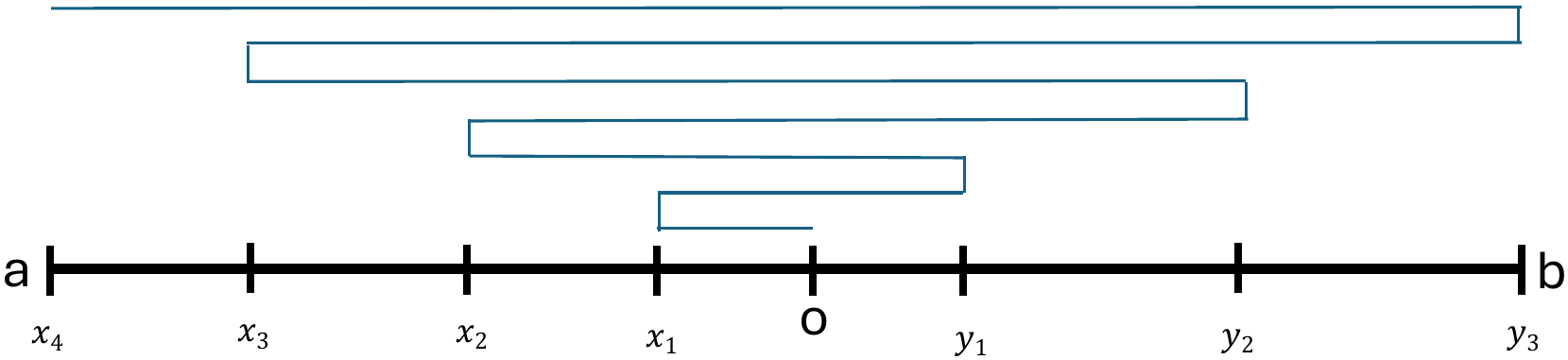}
    \caption{An illustration of $OFF\_TOUR$ for offline TRP or $PRED\_TOUR$ of predicted locations for online TRP, computed using $AFRATI\_ALG$.  The points on $\mathcal{L}$ where the tour changes direction are denoted by $x_i$ on $\mathcal{L}_{left}$ and $y_i$ on $\mathcal{L}_{right}$. }
    \label{fig:tour}
\end{figure}
\begin{figure}[t!]
    \centering
    \includegraphics[width=\linewidth]{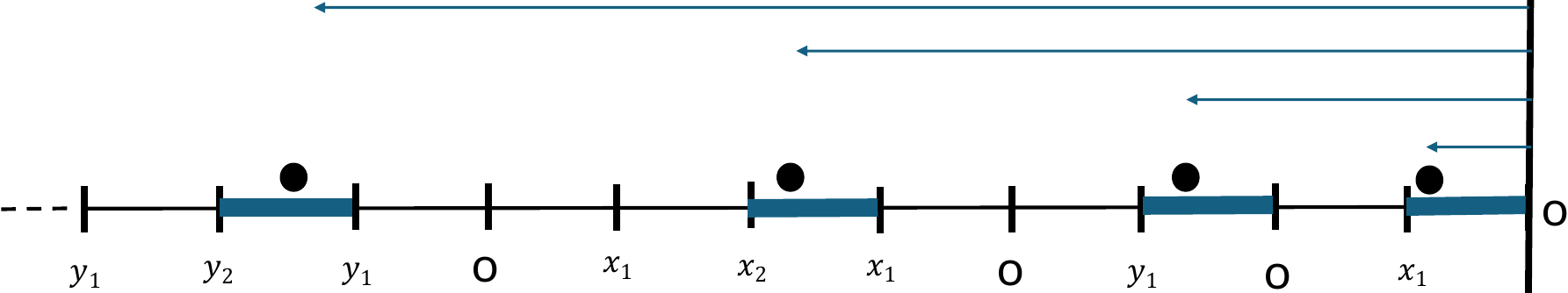}
    \caption{An alternative illustration of $OFF\_TOUR$  or $PRED\_TOUR$ of Fig.~\ref{fig:tour} as a half-line starting from $o$ at the right end (despite it visits $o$ multiple times). Consider any request $r_i$ (denoted as black circle) that arrives at a location on the bold intervals. We have that $OPT(r_i)\geq \{|OFF\_TOUR_i|,t_i\}$, where $|OFF\_TOUR_i|$ is the length of $OFF\_TOUR$ starting from $o$ (right end) to the position $p_i$ of $r_i$ (shown as a arrow line).}
    \label{fig:tour2}
\end{figure}
In the prediction model, given the predicted locations $p^i_{loc}$ of the requests in $\mathcal{R}$, we compute a $PRED\_TOUR$ of the predicted locations applying $AFRATI\_ALG$.
In the prediction model, since requests arrive over time, traversing $PRED\_TOUR$ once may be not enough to serve requests. Additionally, the competitive ratio might become prohibitively large if $PRED\_TOUR$ is not traversed carefully. This is where the second idea helps, that is, use the half-line algorithm (Algorithm \ref{algorithm:half-line}) to visit the $PRED\_TOUR$ so that the competitive ratio can be controlled. 

\vspace{2mm}
\noindent{\bf $AFRATI\_ALG$ \cite{afrati1986complexity}:}
$AFRATI\_ALG$ \cite{afrati1986complexity} is an  offline algorithm meaning that 
the input $\mathcal{R}$ is known beforehand. $AFRATI\_ALG$ computes $OFF\_TOUR$ (Fig.~\ref{fig:tour}) starting from $o$  
picking a direction $left$  or $right$ (from $o$). Suppose $AFRATI\_ALG$ picks direction $left$ first (picking direction $right$ first is analogous). 
Let $x_*$ ($y_*$) denote the point on $\mathcal{L}_{left}$ ($\mathcal{L}_{right}$). Using this notation, the origin $o$ can be written as $(x_0,y_0)$. The leftmost (rightmost) endpoint on $\mathcal{L}_{left}$ ($\mathcal{L}_{right}$) can be written as $x_m=a$ ($y_{m'}=b$). 
The server $s$ visits $\mathcal{L}$ from $o=(x_0,y_0)$ alternating between left and right directions.
Suppose $s$ is currently traversing $left$. Suppose $s$ changes direction at point $x_i$ on $\mathcal{L}_{left}$. After changing direction, it visits  $\mathcal{L}_{right}$ in the $right$ direction. Let $y_i$ be the point on $\mathcal{L}_{right}$ at which $s$  again changes direction. This change of direction continues until either $x_m=a$ is visited after visiting $y_m'=b$ or $y_m'=b$ is visited after visiting $x_m=a$. $[x_i,y_i]$ ( $[y_i,x_{i+1}]$) denotes the segment of $\mathcal{L}$ server $s$ visits in its right (left) direction. 

Let $r_i$ be a request that arrives in the interval $(x_j,x_{j+1}]$  or $(y_j,y_{j+1}]$ (with $j\geq 0$) of $OFF\_TOUR$ with $x_0=y_0=o$. Fig.~\ref{fig:tour2} provides an illustration. $AFRATI\_ALG$ guarantees that $OPT(r_i)=|OFF\_TOUR_i|$, where $|OFF\_TOUR_i|$  is the length of $OFF\_TOUR$ from $o$ until the position $p_i$ of $r_i$. In offline TRP, there is no request to serve between two subsequent bold intervals. 
We have the following theorem. 

\begin{theorem}[\cite{afrati1986complexity}]
$AFTRATI\_ALG$ is optimal  (competitive ratio 1) for  offline TRP on line $\mathcal{L}$. 
\end{theorem}

\begin{algorithm}[t!]
\footnotesize
\begin{algorithmic}[1]
\STATE given the predicted locations of the requests in $\mathcal{R}$, compute $PRED\_TOUR$ of the predicted locations using $AFRATI\_ALG$.
\STATE suppose $PRED\_TOUR$ visits $\mathcal{L}_{left}$ first (visiting $\mathcal{L}_{right}$ first is analogous).
\STATE $\{o,x_1,y_1,x_2,y_2, \ldots, x_m,y_{m'}\}\leftarrow$ $PRED\_TOUR$ such that $x_i,y_i$ are the points at which $PRED\_TOUR$  changes direction, with $x_m=a$ and $y_{m'}=b$. 

\STATE Let $|x_i| (|y_i|$) denote the distance of point $x_i (y_i)$ from $o$. We have that $|x_{i+1}|>|x_{i}|$ and $|y_{i+1}|>|y_{i}|$.

\WHILE{not all requests in $\mathcal{R}$ are served}
\STATE visit $PRED\_TOUR$ in roundtrips $RT_j,j\geq 1$, with each round trip starting from and ending at $o$.
\ENDWHILE

\caption{Algorithm in line  $\mathcal{L}$ with perfect predictions}
 \label{algorithm:line}
 \end{algorithmic}
\end{algorithm}


\subsection{Line Algorithm with Perfect Predictions}
Our algorithm $ON$ is outlined in Algorithm \ref{algorithm:line}. 
Let $\{p_1^{loc}, p_2^{loc},\ldots, p_n^{loc}\}$ be the predicted locations of the $n$ requests in $\mathcal{R}$. 
We run $AFRATI\_ALG$ on the predicted locations and compute a tour (Fig.~\ref{fig:tour}) which we call $PRED\_TOUR:=\{o,x_1,y_1,x_2,y_2,\ldots, x_m,y_{m'}\}$, where $x_i,y_i, i\geq 1$ are the points on $\mathcal{L}_{left}$ and $\mathcal{L}_{right}$, respectively, on which the tour changes direction. (Note that we assume first direction is $left$ from $o$; the case of first direction $right$ is analogous.) 
$PRED\_TOUR$ is now visited in round-trips like Algorithm \ref{algorithm:half-line}. 
At some round trip $RT_k$, there is no request remains to be served and at that time Algorithm \ref{algorithm:line} stops. 

We now analyze Algorithm \ref{algorithm:line} for its correctness and competitive ratio guarantees. Correctness means that eventually all requests in $\mathcal{R}$ are served. 
We need some observations.

\begin{observation}
Consider two consecutive segments $[x_i,y_i]$ and $[y_i,x_{i+1}]$ with $i\geq 1$ (or $[y_j,x_{j+1}]$ and $[x_{j+1},y_{j+1}]$ with $j\geq 0$ and $y_j=o$ for the first segment $[o,x_1]$) of $PRED\_TOUR$. The length of the segment $|y_ix_{i+1}|>|x_iy_i|$ (or $|x_{j+1}y_{j+1}|>|y_jx_{j+1}|$. 
\end{observation}

\begin{observation}
Consider two consecutive segments $[x_i,y_i]$ and $[y_i,x_{i+1}]$ with $i\geq 1$ (or $[y_j,x_{j+1}]$ and $[x_{j+1},y_{j+1}]$ with $j\geq 0$ and $y_j=o$ for the first segment $[o,x_1]$) of $PRED\_TOUR$. Segment $[y_i,x_{i+1}]$ subsumes segment $[x_i,y_i]$, meaning that $[x_i,y_i]$ is completely contained in $[y_i,x_{i+1}]$ (same for $[y_j,x_{j+1}]$ and $[x_{j+1},y_{j+1}]$).
\end{observation}
\begin{observation}
The origin $o$ is contained in  each segment $[x_i,y_i]$ and $[y_i,x_{i+1}]$ of $PRED\_TOUR$. 
\end{observation}

\begin{theorem}
    Algorithm \ref{algorithm:line} correctly serves all the predicted requests in $\mathcal{R}$ on a line $\mathcal{L}$ solving online TRP. 
\end{theorem}
\begin{proof}
From Observations 1, 2, and 3, we have that in each subsequent segment, Algorithm \ref{algorithm:line} traverses more length of $\mathcal{L}$ than the length covered in the previous segment. Furthermore, Algorithm \ref{algorithm:line} does not terminate before both endpoints $a,b$ of $\mathcal{L}$ are visited. If there are still requests in $\mathcal{R}$ remain to be served, Algorithm \ref{algorithm:line} round trips between $a$ and $b$ until all requests in $\mathcal{R}$ are served. 
\end{proof}

We now analyze Algorithm \ref{algorithm:line} for its competitive ratio. 


\begin{observation}
    Suppose each request $r_i$ predicted to arrive at $p^i_{loc}$ is released at time $0$ with no prediction error, traversing $PRED\_TOUR$ only once optimally serves TRP (i.e., competitive ratio 1). 
\end{observation}

\begin{theorem}
\label{theorem:line}
Algorithm \ref{algorithm:line} is $2+\sqrt{3}$-competitive for online TRP on $\mathcal{L}$ in the prediction model with perfect predictions.
\end{theorem}
\begin{proof}
Consider $PRED\_TOUR$ as a half-line $\mathcal{L}$ as shown in Fig.~\ref{fig:tour2}.  Like half-line $\mathcal{L}$, following $PRED\_TOUR$ is optimal even for offline TRP since the predictions are perfect and requests arrive at those locations.  
In the offine TRP, there is no request to serve in $PRED\_TOUR$ between its two consecutive bold intervals (Fig.~\ref{fig:tour2}. In online TRP, there may be requests arriving at time steps such that they need to be served. However, those requests are served when the server traversing currently right (or left) direction changes the direction to traverse left (or right) before it changes again to right (or left).  
Therefore, the competitive ratio of Algorithm \ref{algorithm:half-line} on the half-line $\mathcal{L}$ extends to $PRED\_TOUR$ giving the same  $2+\sqrt{3}$ competitive ratio for online TRP on a line.
\end{proof}

\section{$\min\{2+\sqrt{3}+4\delta,4\}$-Competitive Algorithm on a Line with Imperfect Predictions}
\label{section:upper-line-d}

We assume that $\Delta$ is known. $|\mathcal{L}|$ is known since $\mathcal{L}$ is known. If $\Delta\geq 0.067 |\mathcal{L}|$ (i.e., $\delta\geq 0.067$), we run the algorithm of \cite{BienkowskiKL21-ICALP} in the original model. 
However, if $\Delta< 0.067 |\mathcal{L}|$ (i.e., $\delta<0.067)$,  we modify Algorithm \ref{algorithm:line} as follows.   
Given the predicted locations $\{p_1^{loc}, p_2^{loc},\ldots, p_n^{loc}\}$, we get new predicted locations $\{\hat{p}_1^{loc}, \hat{p}_2^{loc},\ldots, \hat{p}_n^{loc}\}$ setting $\hat{p}_i^{loc}=p_i^{loc}-\Delta$. We then use  $AFRATI\_ALG$ and compute $PRED\_TOUR'$ of locations $\{\hat{p}_1^{loc}, \hat{p}_2^{loc},\ldots, \hat{p}_n^{loc}\}$. 
We have the following observation for $PRED\_TOUR'$.

\begin{observation}
    There does not exist any other solution which gives better  sum of completion times than $PRED\_TOUR'$ on a line with predictor error $\Delta$. 
\end{observation}

Since we have error $\Delta$, we modify the tour $PRED\_TOUR'$ to obtain $PRED\_TOUR''$. We then run  Algorithm \ref{algorithm:line} on $PRED\_TOUR''$ with error $\Delta$. 
Let $\{o,x_1,y_1,x_2,y_2, \ldots, x_m,y_{m'}\}$ be  $PRED\_TOUR'$ such that $x_i,y_i$ are the points at which it  changes direction.  Except $o$, we shift each point $x_i,y_i$ in $PRED\_TOUR'$ adding $2\Delta$. Doing so, there might be less turning points in $PRED\_TOUR''$ than $PRED\_TOUR'$ but it will not affect our our algorithm.  
We then visit  $PRED\_TOUR''$ like Algorithm \ref{algorithm:line} in round-trips with each round-trip  length as follows:
the first round-trip $RT_1$ is of length $2+2\alpha+4\Delta$ and for $j\geq 2$, $|RT_j|=(2+2\alpha)^{j-1}(1+2\alpha)+ 4\Delta$. Fig.~\ref{fig:deltaroundtrip} provides an illustration of the adjusted round-trips given $\Delta$. The correctness is immediate as in Algorithm \ref{algorithm:line}. We prove the following theorem for the competitive ratio guarantee.

\begin{figure}[t!]
    \centering
    \includegraphics[width=\linewidth]{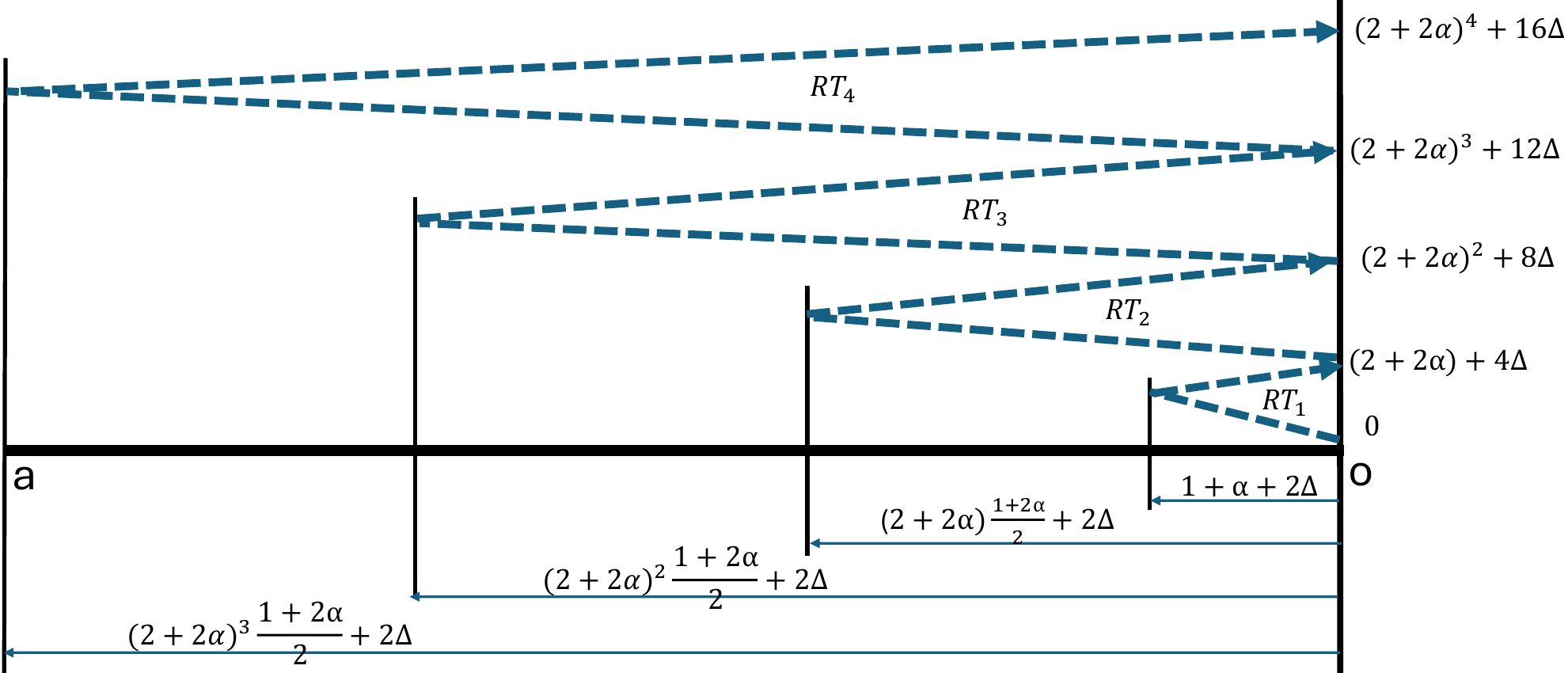}
    \caption{An illustration of how round-trips in Fig.~\ref{fig:halfline} are  adjusted to accommodate error $\Delta$ in predicted locations.}
    \label{fig:deltaroundtrip}
\end{figure}
\begin{theorem}
    Algorithm \ref{algorithm:line} modified with $\Delta$ is $\min\{2+\sqrt{3}+4\delta,4\}$-competitive for online TRP on a line $\mathcal{L}$, where $\delta=\frac{\Delta}{|L|}$.
\end{theorem}
\begin{proof}
For $\Delta\geq 0.067 |\mathcal{L}|$ ($\delta\geq 0.067$), 4-competitive ratio is obtained by the algorithm of \cite{BienkowskiKL21-ICALP}. 

For $\Delta< 0.067 |\mathcal{L}|$ ($\delta< 0.067$), our analysis of Algorithm \ref{algorithm:line} with error $\Delta$ is as follows. 
Consider a request $r_i=(p_i,t_i)$. Let $|PRED\_TOUR''_i|$ be the length of $PRED\_TOUR''$ starting from $o$ that visits $p_i$ for the first time. Let $|PRED\_TOUR'_i|$ be the length of $PRED\_TOUR'$ starting from $o$ that visits $p_i$ for the first time. 

We have two cases depending on whether  $|PRED\_TOUR'_i|\geq t_i$ or $|PRED\_TOUR'_i|< t_i$. Let $r_i$ be visited for the first time in $RT_k, k\geq 1$. 

\begin{itemize}
\item {\bf Case 1:} $|PRED\_TOUR'_i|\geq t_i$. According to the analysis analogous to Theorem \ref{theorem:half-line},  
$$ON(r_i)\leq |PRED\_TOUR'_i|+(2+2\alpha)^{k-1}+(k-1)4\Delta,$$ if when $p_i$ is visited for the first time in $RT_k$.
Additionally,  
$|PRED\_TOUR'_i|\geq \frac{(2+2\alpha)^{k-2}(1+2\alpha)}{2}$ when $p_i$ is visited for the first time in $RT_k$. 
Therefore, $$\frac{ON(r_i)}{OPT(r_i)}=\frac{\frac{(2+2\alpha)^{k-2}(1+2\alpha)}{2}+(2+2\alpha)^{k-1}+(k-1)4\Delta}{\frac{(2+2\alpha)^{k-2}(1+2\alpha)}{2}},$$
which reduces to 
$\frac{ON(r_i)}{OPT(r_i)}=\frac{5+6\alpha}{1+2\alpha}+\frac{(k-1)8\Delta}{(2+2\alpha)^{k-2}(1+2\alpha)}.$

\item {\bf Case 2:} $|PRED\_TOUR'_i|< t_i$.
As in Theorem \ref{theorem:half-line},
if $r_i$'s location is in on the segment of $\mathcal{L}$ visited by $PRED\_TOUR'_i$,  $\frac{ON(r_i)}{OPT(r_i)}= 2+2\alpha+4\Delta.$  
If $|PRED\_TOUR'_i|\geq  t_i$, then Case 1 applies as  in Theorem \ref{theorem:half-line} which is $\frac{ON(r_i)}{OPT(r_i)}=\frac{5+6\alpha}{1+2\alpha}+\frac{(k-1)8\Delta}{(2+2\alpha)^{k-2}(1+2\alpha)}.$
\end{itemize}

Combining the competitive ratios for Cases 1 and 2,
$$\frac{ON(r_i)}{OPT(r_i)}=\max\{2+2\alpha+4\Delta,\frac{5+6\alpha}{1+2\alpha}+\frac{(k-1)8\Delta}{(2+2\alpha)^{k-2}(1+2\alpha)}\}.$$

    Setting $\alpha=\frac{\sqrt{3}}{2}$, for any $k\geq 1$, 
$\frac{(k-1)8\Delta}{(2+2\alpha)^{k-2}(1+2\alpha)}\leq \frac{8}{1+\sqrt{3}}\Delta < 4 \Delta.$
Therefore,  since $\delta=\frac{\Delta}{|\mathcal{L}|}$,
$\frac{ON(r_i)}{OPT(r_i)}=2+\sqrt{3}+4\delta.$
We have the claimed bound since if 
the algorithm of \cite{BienkowskiKL21-ICALP} is applied, Algorithm \ref{algorithm:line} is not applied (and vice-versa).
\end{proof}

\section{Concluding Remarks}
\label{section:conclusion}
In this paper, we have studied online TRP  in the prediction model  and established a $\min\{3.732+4\delta,4\}$-competitive deterministic algorithm, with any error $0\leq \delta\leq 1$, which is an improvement compared to the state-of-the-art deterministic bound of 4 in the original model. We also established a lower bound of 3 improving on the state-of-the-art $1+\sqrt{2}$. For future work, it would be interesting to extend the competitive ratio to any metric satisfying triangle inequality (not just line). It would also be interesting to establish stronger lower bound. Finally, it would be interesting to remove the knowledge of $\delta$ from our algorithm, consider other appropriate prediction models, and perform experimental evaluation.

\bibliographystyle{named}
\bibliography{references}


\end{document}